\documentclass[a4paper,twoside]{article}

\usepackage{epsfig}
\usepackage{subcaption}
\usepackage{calc}
\usepackage{amstext}
\usepackage{amsthm}
\usepackage{multicol}
\usepackage{multirow}
\usepackage{pslatex}
\usepackage{apalike}
\usepackage{hyperref}
\usepackage[bottom]{footmisc}    

\usepackage{tikz}
\usetikzlibrary{shapes.gates.logic.US,trees,positioning,arrows}
\usepackage{pgfplots}

\usepackage{graphicx}
\usepackage{amsmath,amsfonts,amssymb}
\usepackage[lined,boxed,linesnumbered]{algorithm2e}
\usepackage{SCITEPRESS} 

\newtheorem{definition}{Definition}[section]
\newtheorem{problem}[definition]{Problem}
\newtheorem{theorem}[definition]{Theorem}
\newtheorem{lemma}[definition]{Lemma}
\theoremstyle{definition}

\newtheorem{example}[definition]{Example}

\newcommand{\recht}[1]{\operatorname{#1}}

\newcommand{\pp}{\mathrm{p}}
\newcommand{\tOR}{\mathtt{OR}}
\newcommand{\tAND}{\mathtt{AND}}

\newcommand{\tBE}{\mathtt{BE}}

\newcommand{\BE}[1]{\recht{BE}_{#1}}

\newcommand{\struc}[1]{\recht{S}_{#1}}
\newcommand{\ch}{\recht{ch}}
\newcommand{\RRR}[1]{\recht{R}_{#1}}

\newcommand{\CS}[1]{\recht{CS}_{#1}}

\newcommand{\tzero}{\mathtt{0}}
\newcommand{\tone}{\mathtt{1}}
\newcommand{\BB}{\mathbb{B}}
\newcommand{\prob}{\mathbb{P}}
\newcommand{\idom}{\recht{id}}

\newcommand{\SFPA}{\mathcal{A}}
\newcommand{\FF}[1]{\mathsf{L}_{#1}}
\newcommand{\mon}[1]{\langle#1\rangle}
\newcommand{\SQFU}{\mathtt{SFPA}}
\newcommand{\SQFUU}{\mathtt{SFPA2}}

\newcommand{\tCBE}{\mathtt{CBE}}
\newcommand{\CBE}[1]{\recht{CBE}_{#1}}
\newcommand{\RR}{\mathbb{R}}

\allowdisplaybreaks

\begin{document}

\pagestyle{plain}

\title{Fault tree reliability analysis via squarefree polynomials}

\author{\authorname{Milan Lopuhaä-Zwakenberg}
\affiliation{University of Twente, Enschede, the Netherlands}
\email{m.a.lopuhaa@utwente.nl}
}


\abstract{Fault tree (FT) analysis is a prominent risk assessment method in industrial systems. Unreliability is one of the key safety metrics in quantitative FT analysis. Existing algorithms for unreliability analysis are based on binary decision diagrams, for which it is hard to give time complexity guarantees beyond a worst-case exponential bound. In this paper, we present a novel method to calculate FT unreliability based on algebras of squarefree polynomials and prove its validity. We furthermore prove that time complexity is low when the number of multiparent nodes is limited. Experiments show that our method is competitive with the state-of-the-art and outperforms it for FTs with few multiparent nodes.}

\onecolumn \maketitle \normalsize \setcounter{footnote}{0} \vfill

\section{Introduction}

\noindent \textit{Fault trees}. Fault trees (FTs) form a prominent risk assessment method to categorize safety risks on industrial systems. A FT is a hierarchical graphical model that shows how failures may propagate and lead to system failure. Because of its flexibility and rigor, FT analysis is incorporated in many risk assessment methods employed in industry, including FaultTree+ \cite{FT+} and TopEvent FTA \cite{TopEventFTA}.	

A FT is a directed acyclic graph (not necessarily a tree) whose root represents system failure. The leaves are called \emph{basic events} (BEs) and represent atomic failure events. Intermediate nodes are AND/OR-gates, whose activation depends on that of their children; the system as a whole fails when the root is activated. An example is given in Fig.~\ref{fig:ex}.

\noindent \textit{Quantitative analysis}. Besides a qualitative analysis of what sets of events cause overall system failure, FTs also play an important role in \emph{quantitative risk analysis}, which seeks to express the safety of the system in terms of safety metrics, such as the total expected downtime, availability, etc. An important safety metric is \emph{(un)reliability}, which, given the failure probability of each BE, calculates the probability of system failure. As the size of FTs can grow into the hundreds of nodes \cite{ruijters2019ffort}, calculating the unreliability efficiently is crucial for giving safety and availability guarantees.

\begin{figure} 
\centering
\includegraphics[width=6cm]{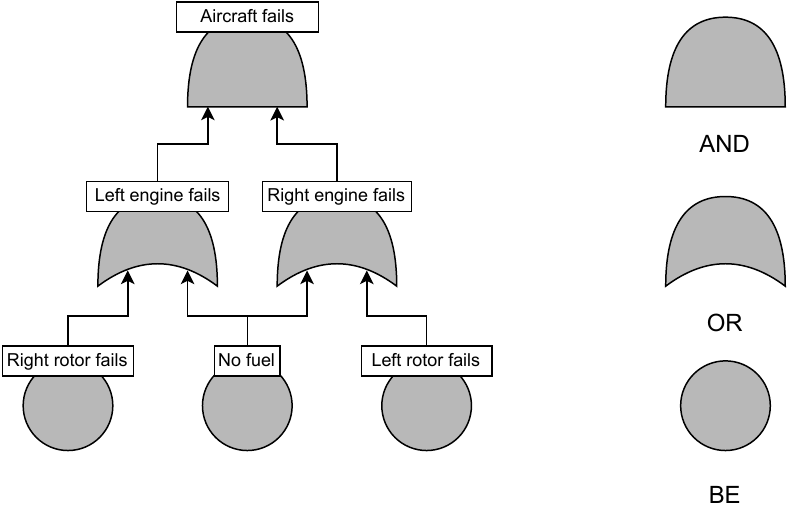}
\caption{A fault tree for a small aircraft. The aircraft fails if both its engines fail; each engine fails if either its rotor fails or it has no fuel (the plane has a single fuel tank).} \label{fig:ex}
\end{figure}

There exist two main approaches to calculating unreliability \cite{ruijters2015fault}. The first approach works bottom-up, recursively calculating the failure probability of each gate. This algorithm is fast (linear time complexity), but only works as long as the FT is actually a tree. However, nodes with multiple parents (DAG-like FTs) are necessary to model more intricate systems. For such FTs, as we show in this paper, calculating unreliability is NP-hard. The main approach for such FTs is based on translating the FT into a binary decision diagram (BDD) and performing a bottom-up analysis on the BDD. This BDD is of worst-case exponential size, though heuristics exist. The BDD corresponding to an FT depends on a linear ordering of the BEs, with different orderings yielding BDDs of wildly varying size; although any single one of them can be used to calculate unreliability, finding the optimal BE ordering is an NP-hard problem in itself. As a result, it is hard to give guarantees on the runtime of this unreliability calculation algorithm in terms of properties of the FT.

\noindent \textit{Contributions.} In this paper, we present a radical new way for calculating unreliability for general FTs. The bottom-up algorithm does not work for DAG-like FTs, since it does not recognize multiple copies of the same node in the calculation, leading to double counting. In our approach we amend this by keeping track of nodes with multiple parents, as these may occur twice in the same calculation. Then, instead of propagating failure probabilities as real numbers, we propagate squarefree polynomials whose variables represent the failure probabilities of nodes with multiple parents; keeping these formal variables allows us to detect and account for double counting. Furthermore, to keep complexity down we replace formal variables with real numbers whenever we are able. At the root all formal variables have been substituted away, yielding the unreliability as a real number.

This approach has as advantage over BDD-based algorithms that we are able to give upper bounds to computational complexity. Most of the complexity comes from the fact that we do arithmetic with polynomials rather than real numbers. However, if the number of multiparent nodes is limited, these polynomials have limited degree, and computation is still fast. We prove this formally, by showing that the time complexity is linear when the number of multiparent nodes is bounded, and in experiments, in which we compare our method to Storm-dft \cite{basgoze2022bdds}, a state-of-the-art tool for FT analysis using a BDD-based approach: here our method is competitive in general and is considerably faster for FTs with few multiparent nodes.

Summarized our contributions are the following:
\begin{enumerate}
\item A new algorithm for fault tree reliability analysis based on squarefree polynomial algebras;
\item A proof of the algorithm's validity and bounds on its time complexity;
\item Experiments comparing our algorithm to the state-of-the-art.
\end{enumerate}

An artefact of the experiments is available at \cite{code}.

\section{Related work}

There exists a considerable amount of work on FT reliability analysis. FTs were first introduced in \cite{watson1961launch}. A bottom-up algorithm to calculate their reliability is presented in \cite{ruijters2015fault}, which is a formalization of mathematical principles that have been used since the beginning of FT analysis. This algorithm works only for FTs that are actually trees, i.e., do not contain nodes with multiple parents. In \cite{rauzy1993new}, a BDD-based method for calculating reliability was introduced that also works for DAG-shaped FTs. This method is still the state of the art, although some improvements have been made, notably by tweaking variable ordering to obtain smaller BDDs \cite{bouissou1997bdd} and dividing the FT into so-called modules that can be handled separately \cite{rauzy1997exact}.

FTs are also used to model the system's reliability over time. The failure time of each BE is modeled as a random variable (typically exponentially distributed), and additional gate types are introduced to model more elaborate timing behavior. Reliability analysis of these \emph{dynamic FTs} is done via stochastic model checking \cite{basgoze2022bdds}.

\section{Fault trees}

In this section, we give the formal definition of fault trees (FTs) and their reliability used in this paper. For us FTs are static (only AND/OR gates), and each basic event $v$ is assigned a failure probability $\pp(v)$. Thus we define:

\begin{definition}
A \emph{fault tree} (FT) is a tuple $T = (V,E,\gamma,\pp)$ where:
\begin{itemize}
    \item $(V,E)$ is a rooted directed acyclic graph;
    \item $\gamma$ is a function $\gamma\colon V \rightarrow \{\tOR,\tAND,\tBE\}$ such that $\gamma(v) = \tBE$ if and only if $v$ is a leaf;
    \item $\pp$ is a function $\pp\colon \BE{T} \rightarrow [0,1]$, where $\BE{T} = \{v \in V \mid \gamma(v) = \tBE\}$.
\end{itemize}
\end{definition}

Note that a FT is not necessarily a tree as gates may share children. The root of $T$ is denoted $\RRR{T}$. For a node $v$, we let $\ch(v)$ be the set of children of $v$.

The \emph{structure function} determines, given a gate and a safety event, whether the event succesfully propagates to the gate. Here we model a safety event as the set of BEs happening, which can be encoded as a binary vector $\vec{f} \in \BB^{\BE{T}}$, where $f_v = \tone$ denotes that the BE $v$ occurs in the event. The structure function is then defined as follows.

\begin{definition}
Let $T = (V,E,\gamma,\pp)$ be a FT.
\begin{enumerate}
    \item A \emph{safety event} is an element of $\BB^{\BE{T}}$.
    \item The \emph{structure function} of $T$ is the function $\struc{T}\colon V \times \BB^{\BE{T}} \rightarrow \BB$ defined recursively by
    \[
\struc{T}(v,\vec{f}) = \begin{cases}
    f_v, & \textrm{ if $\gamma(v) = \tBE$,}\\
    \bigvee_{w \in \ch(v)} \struc{T}(w,\vec{f}), & \textrm{ if $\gamma(v) = \tOR$},\\
    \bigwedge_{w \in \ch(v)} \struc{T}(w,\vec{f}), & \textrm{ if $\gamma(v) = \tAND$}.
\end{cases}
    \]
    \item A safety event $\vec{f}$ such that $\struc{T}(\RRR{T},\vec{f}) = \tone$ is called a \emph{cut set}; the set of all cut sets is denoted $\CS{T}$.
\end{enumerate}
\end{definition}

Quantitative analysis of a FT is typically done via its \emph{unreliability}, i.e., the probability of a cut set occurring, where each BE $v$ has probability $\pp(v)$ of happening. The BE failure probabilities are considered to be independent. The reasoning behind this is that when they are not independent, this is due to some common cause; this common cause should then be explicitely modeled in the FT framework, by replacing the two non-independent BEs with sub-FTs that share common nodes \cite{pandey2005fault}.

\begin{definition} \label{def:unreliability}
Let $T = (V,E,\gamma,\pp)$ be a FT. Let $\vec{F} \in \BB^{{\BE{T}}}$ be the random variable defined by $\prob(F_v = \tone) = \pp(v)$ for all $v \in \BE{T}$, and all these events are independent. Then the \emph{unreliability} of $T$ is defined as
\begin{align*}
U(T) &= \prob(\vec{F} \in \CS{T}) \\
&= \sum_{\vec{f} \in \CS{T}} \prod_{v\colon f_v = \tone} \pp(v)\prod_{v\colon f_v = \tzero} (1-\pp(v)).
\end{align*}
\end{definition}

\begin{example}
Consider the FT $T$ from Fig.~\ref{fig:ex}. Abbreviating BE names, assume $\pp(\mathsf{rrf}) = \pp(\mathsf{lrf}) = 0.4$ and $\pp(\mathsf{nf}) = 0.3$. Furthermore, write $\vec{f} \in \BB^{\BE{T}}$ as $f_{\mathsf{rrf}}f_{\mathsf{nf}}f_{\mathsf{lrf}}$. Then $\CS{T} = \{010,011,101,110,111\}$, so
\begin{align*}
U(T) &= 0.6 \cdot 0.3 \cdot 0.6 \\
&+ 0.6 \cdot 0.3 \cdot 0.4 \\
&+ 0.4 \cdot 0.7 \cdot 0.4 \\
&+ 0.4 \cdot 0.3 \cdot 0.6 \\
&+ 0.4 \cdot 0.3 \cdot 0.4 = 0.412.
\end{align*}
\end{example}

The unreliability $U(T)$ represents the probability of failure of the system modeled by the fault tree and is crucuial to providing safety and availability guarantees. The expression in Definition \ref{def:unreliability} becomes too large to handle for large FTs very quickly; thus it is important to find efficient solutions to the following problem.

\begin{problem} \label{prob:ut}
Given a FT $T$, calculate $U(T)$.
\end{problem}

Unfortunately, we show that this  problem is NP-hard. The reason for this is that with appropriately chosen probabilities, one can find from $U(T)$ a minimal element of $\CS{T}$, and finding such a so-called \emph{minimal cut set} is known to be NP-hard \cite{rauzy1993new}.

\begin{theorem} \label{thm:NP}
Problem \ref{prob:ut} is NP-hard.
\end{theorem}

\subsection{Existing $U(T)$ algorithms}

There are two prominent algorithms for calculating $U(T)$. The first one calculates, for each node $v$, the probability $g_v = \prob(\struc{T}(v,\vec{F}) = \tone)$ bottom-up \cite{ruijters2015fault}. For BEs one has $g_v = \pp(v)$. For an AND-gate, one has $g_v = \prod_{w \in \ch(v)} g_w$ as long as the events $\struc{T}(w,\vec{F}) = \tone$ are independent as $w$ ranges over all children of $v$. This happens when no two children of $v$ have any shared descendants. For OR-gates one likewise has $g_v = 1-\prod_{w \in \ch(v)} g_w$. This gives rise to a linear-time algorithm that calculates $g_v$ bottom-up. Unfortunately, this algorithm only works for FTs that have a tree structure: as soon as a node has multiple parents, the independence assumption will be violated at some point in the calculation.

The second algorithm \cite{rauzy1993new} works for general FTs, and works by translating the Boolean function $\struc{T}(\RRR{T},-)$ into a \emph{binary decision diagram}, which is a directed acyclic graph, in which the evaluation of the function at a boolean vector is represented by a path through the graph. After the BDD is found, $U(T)$ can be calculated using a bottom-up algorithm on the BDD, whose time complexity is linear in the size of the BDD. Unfortunately, the size of the BDD is worst-case exponential, although this worst case is seldomly attained in practice \cite{rauzy1997exact,bobbio2013methodology}. To construct the BDD, one first has to linearly order the variables; finding the order that minimizes BDD size is an NP-hard problem, although heuristics exist \cite{valiant1979complexity,le2014novel}.  

\section{An example of our method}

\begin{figure}
    \centering
    \includegraphics[width=4cm]{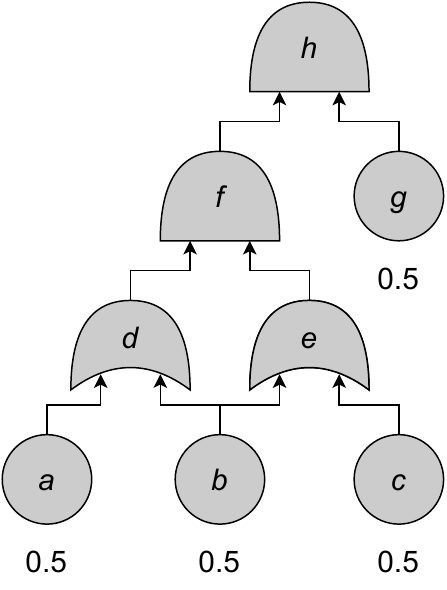}
    \caption{An example FT with failure probabilities.}
    \label{fig:algo_ex}
\end{figure}

Before we dive into the details of our method, we go through an example to showcase how the method works and to motivate the technical sections. 

Consider the FT of Fig.~\ref{fig:algo_ex}. In a bottom-up method, we calculate the failure probability $g_v$ of each node $v$: thus for the BEs we have $g_a = g_b = g_c = g_g = 0.5$. For $d$, the bottom-up method dictates that we should calculate $g_d = g_a+g_b-g_ag_b$. However, this will cause problems at $f$, since $d$ and $e$ share $b$ as child, and so their failure probabilities will not be independent. Thus, at $d$, we modify $g_d$ to `remember' its dependence on $b$. We do so by introducing a formal variable $\FF{b}$ representing $b$'s failure probability, yielding $g_d = g_a+\FF{b}-g_a\FF{b} = 0.5+0.5\FF{b}$. We also get $g_e = 0.5+0.5\FF{b}$. Note that we only introduce a formal variable for $b$, and not for $a$ and $c$, as the latter only have one parent node and thus have no chance of appearing twice in the same calculation. 

At $f$, we calculate $g_f = g_dg_e = 0.25+0.5\FF{b}+0.25\FF{b}^2$. Here we introduce the rule $\FF{b}^2 = \FF{b}$, so that $g_f = 0.25+0.75\FF{b}$. The idea behind this is that we use multiplication to determine the possibility of two events occurring simultaneously. Since $\FF{b}$ represents $\prob(\struc{T}(b,\vec{F}) = \tone)$, the term $\FF{b}^2$ actually represents $\prob(\struc{T}(b,\vec{F}) = \tone \wedge \struc{T}(b,\vec{F}) = \tone)$. This is equal to $\prob(\struc{T}(b,\vec{F}) = \tone)$, so $\FF{b}^2=\FF{b}$.

At this point, the graph structure tells us that $b$ cannot appear twice in the same calculation any more. Thus we can safely substitute $g_b=0.5$ for $\FF{b}$ in $g_f = 0.25+0.75\FF{b}$, yielding $g_f = 0.625$. Finally, we get $g_h = g_fg_g = 0.3125$.

In the following sections, we introduce two mathematical tools needed to apply this method in greater generality. In Section \ref{sec:dom} we review the graph-theoretic notion of dominators, which will tell us for which nodes we need to introduce formal variables and at what point they can be substituted away. In Section \ref{sec:sfpa} we formalize the polynomial algebra in which our arithmetic takes place.

\section{Preliminaries I: Dominators} \label{sec:dom}

In this section we review the concept of dominators and apply them to FTs. We need this to determine at what point in the bottom-up calculation, outlined in the previous section, we can replace a formal variable $\FF{v}$ with the expression $g_v$. Informally, a dominator of $v$ is present on all paths from the root to $v$.

\begin{definition} \emph{\cite{prosser1959applications}}
Let $T = (V,E,\gamma,\pp)$ be a FT.
\begin{enumerate}
\item Define a partial order $\preceq$ on $V$ by $x \preceq y$ iff there is a path $y \rightarrow x$ in $T$.
\item Given two nodes $v,w \in V$, we say that $w$ \emph{dominates} $v$ if $v \prec w$ and every path $\RRR{T} \rightarrow v$ in $T$ contains $w$.
\end{enumerate}
\end{definition}

The set of dominators of a node is nonempty and has a minimum:

\begin{lemma} \emph{\cite{lengauer1979fast}}
If $v \neq \RRR{T}$, then there is a unique $w$ dominating $v$ such that each $w'$ dominating $v$ satisfies $w \preceq w'$; this $w$ is called the \emph{immediate dominator of $v$}, denoted $w = \idom(v)$. \qed
\end{lemma}

\begin{example}
In Fig.~\ref{fig:algo_ex}, the dominators of $a$ are $d$, $f$, and $h$, and $\idom(a) = d$. The dominators of $b$ are $f$ and $h$, and $\idom(b) = f$. Note that $d$ is not a dominator of $b$, since the path $h \rightarrow f \rightarrow e \rightarrow b$ does not pass through $d$.
\end{example}

The immediate dominator is interesting to us since, as we will see later, at $\idom(v)$ we can replace $\FF{v}$ by $g_v$. The following result relates the relative position of $v$ and $w$ to that of their immediate dominators.

\begin{lemma} \label{lem:idom}
If $v \prec w$, then either $\idom(v) \preceq w$ or $\idom(w) \preceq \idom(v)$.
\end{lemma}

To use these definitions in an algorithmic context, we will make use of the following result:

\begin{theorem} \emph{\cite{lengauer1979fast}} \label{thm:idom}
Given $T$, there exists an algorithm of time complexity $\mathcal{O}(|E|)$ that finds $\idom(v)$ for each $v \in V$. \qed
\end{theorem}

\section{Preliminaries II: Squarefree polynomial algebras} \label{sec:sfpa}

In this second preliminary section, we formally define the algebras in which our calculations take place. These are similar to multivariate polynomial algebras, except in every monomial every variable can have degree at most $1$.

\begin{definition}
Let $X$ be a finite set. We define the \emph{squarefree real polynomial algebra over $X$} to be the algebra $\SFPA(X)$ consisting of formal sums
\[
\alpha = \sum_{Y \subseteq X} \alpha_Y \prod_{x \in Y} \FF{x},
\]
where the $\FF{x}$ are formal variables and $\alpha_Y \in \mathbb{R}$. Addition and multiplication are as normal polynomials, except that they are subject to the law $\FF{x}^2 = \FF{x}$ for all $x \in X$; that is,
\begin{align*}
(\alpha+\beta)_Y &= \alpha_Y + \beta_Y,\\
(\alpha \cdot \beta)_Y &= \sum_{\substack{Y',Y'' \subseteq X\colon\\Y'\cup Y'' = Y}} \alpha_{Y'}\beta_{Y''}.
\end{align*}
\end{definition}

\begin{example} \label{ex:poly}
Let $X = \{x,y\}$. Let $\alpha = 2+\FF{x}+\FF{y}$ and $\beta = \FF{x}+3\FF{x}\FF{z}$. Coefficientwise, $\alpha$ is described as
\[
\alpha_X = \begin{cases}
2, & \textrm{ if $X = \varnothing$},\\
1, & \textrm{ if $X = \{x\}$ or $X = \{y\}$},\\
0, & \textrm{ if $X = \{x,y\}$.}
\end{cases}
\]
Furthermore, $\alpha+\beta = 2+2\FF{x}+\FF{y}+3\FF{x}\FF{z}$, and
\begin{align*}
\alpha \cdot \beta &= \alpha \cdot \FF{x} + \alpha \cdot (3\FF{x}\FF{z}) \\
&= (2\FF{x}+\FF{x}+\FF{x}\FF{y}) + (6\FF{x}\FF{z}+3\FF{x}\FF{z}+3\FF{x}\FF{y}\FF{z}) \\
&= 3\FF{x}+\FF{x}\FF{y}+9\FF{x}\FF{z}+3\FF{x}\FF{y}\FF{z}.
\end{align*}
\end{example}

Note that if $X \subseteq X'$, an $\alpha \in \SFPA(X)$ can also be considered an element of $ \SFPA(X')$, by taking $\alpha_{Y} = 0$ whenever $Y \not \subseteq X$. For two sets $X$ and $Y$ this also allows us to add and multiply $\alpha \in \SFPA(X)$ and $\beta \in \SFPA(Y)$, by considering both to be elements of $\SFPA(X \cup Y)$. In the rest of this paper we will do this without comment.

Besides multiplication and addition, another important operation that we need is the substitution of a formal variable by a polynomial. This works the same as with regular polynomials.

\begin{definition} \label{def:sub}
Let $X,Y$ be finite sets and $x \in X \setminus Y$. Let $\alpha \in \SFPA(X)$ and $\beta \in \SFPA(Y)$. Then the \emph{substitution} $\alpha[\FF{x} \mapsto \beta]$ is the element of $\SFPA(X\setminus\{x\} \cup Y)$ obtained by replacing all instances of $\FF{x}$ by $\beta$; more formally, $\alpha[\FF{x} \mapsto \beta]$ is expressed as
\begin{align*}
&\beta \cdot \left(\sum_{\substack{Z \subseteq X\colon\\x \in Z}} \alpha_Z \prod_{x' \in Z\setminus\{x\}} \FF{x'}\right) + \sum_{\substack{Z \subseteq X\colon\\x \notin Z}} \alpha_Z \left(\prod_{x' \in Z} \FF{x'}\right).
\end{align*}
In terms of coefficients this is expressed as
\[
\alpha[\FF{x} \mapsto \beta]_Z = \alpha_Z + \sum_{\substack{x \in Z' \subseteq X,\\ Z'' \subseteq Y\colon\\ Z' \setminus \{x\} \cup Z'' = Z}} \alpha_{Z'}\beta_{Z''}
\]
where $\alpha_Z = 0$ if $Z \not \subseteq X$.
\end{definition}

Note that in this definition the multiplication and addition are of elements of $\SFPA(X \setminus \{x\} \cup Y)$. 

\begin{example}
Continuing Example \ref{ex:poly}, the substitution $\beta[\FF{z}\mapsto \alpha]$ is equal to
\begin{align*}
\beta[\FF{z} \mapsto \alpha] &= \FF{x}+3\FF{x} \cdot (2+\FF{x}+\FF{y}) \\
&= 10\FF{x}+3\FF{x}\FF{y}.
\end{align*}
\end{example}

In what follows, we will need three results on calculation in $\mathcal{A}(X)$. The first result considerably simplifies the substitution operation.

\begin{lemma} \label{lem:sub}
Let $x,\alpha,\beta$ be as in Definition \ref{def:sub}. Then
\[
\alpha[\FF{x} \mapsto \beta] = \alpha[\FF{x} \mapsto 1]\cdot \beta + \alpha[\FF{x} \mapsto 0] \cdot (1-\beta).
\]
\end{lemma}

The second result shows how substitution behaves with respect to addition and multiplication.

\begin{lemma} \label{lem:subarit}
Let $\alpha_1,\alpha_2 \in \SFPA(X)$, $\beta \in \SFPA(Y)$, and $x \in X \setminus Y$. Then:
\begin{enumerate}
\item $(\alpha_1+\alpha_2)[\FF{x} \mapsto \beta] = \alpha_1[\FF{x} \mapsto \beta] + \alpha_2[\FF{x} \mapsto \beta]$.
\item If $\beta^2 = \beta$, then furthermore $(\alpha_1\alpha_2)[\FF{x} \mapsto \beta] = \alpha_1[\FF{x} \mapsto \beta] \cdot \alpha_2[\FF{x} \mapsto \beta]$.
\end{enumerate} 
\end{lemma}

The third result states that two substitution operations can be interchanged, as long as one does not substitute a variable present in the other:

\begin{lemma} \label{lem:subexchange}
Let $\alpha \in \SFPA(X)$, $\beta_1 \in \SFPA(Y_1)$, $\beta_2 \in \SFPA(Y_2)$, $x_1,x_2 \in X \setminus (Y_1 \cup Y_2)$. If $x_1 \notin Y_2$ and $x_2 \notin Y_1$, then $\alpha[\FF{x_1}\mapsto \beta_1][\FF{x_2} \mapsto \beta_2] = \alpha[\FF{x_2} \mapsto \beta_2][\FF{x_1} \mapsto \beta_1]$.
\end{lemma}

When this lemma applies and the order of substitutions does not matter, we will write expressions like $\alpha[\FF{x_1} \mapsto \beta_1,\FF{x_2} \mapsto \beta_2]$, or even $\alpha[\forall i \leq n\colon \FF{x_i} \mapsto \beta_i]$.

Note that as an $\mathbb{R}$-algebra, one may identify $\mathcal{A}(X)$ with $K/I$, where $K = \mathbb{R}[\FF{x}\colon x \in X]$ is a free polynomial algebra and $I$ is the ideal generated by the set $\{\FF{x}^2-\FF{x} \mid x \in X\}$. However, the substitution operation does not correspond to a `natural' operation on on $K/I$.

\subsection{Real-valued Boolean functions}

We will use the elements of $\mathcal{A}(X)$ is to represent functions $\BB^X \rightarrow \mathbb{R}$. The following result states that this can be done in a unique way. Since both elements of $\mathcal{A}(X)$ and functions $\BB^X \rightarrow \mathbb{R}$ can be represented by $2^{|X|}$ real numbers, this should come as no surprise.

\begin{theorem} \label{thm:realfunc}
Let $X$ be a finite set, and let $g\colon \BB^X \rightarrow \RR$ be any function. Then there exists a unique $\mon{g} \in \SFPA(X)$ such that $
g(c) = \mon{g}[\forall x \in X\colon \FF{x} \mapsto c_x]$
for all $\vec{c} \in \BB^X$.
\end{theorem}

\begin{example}
Let $X = \{x,y\}$, with $\vec{c} \in \BB^X$ represented as $c_xc_y$. Consider the function $g\colon \BB^X \rightarrow \RR$ given by
\begin{align*}
g(00) &= 3,& g(01) &= -2, \\
g(10) &= 7,& g(11) &= 4.
\end{align*}
Suppose $\mon{g} = k_1+k_2\FF{x}+k_3\FF{y}+k_4\FF{x}\FF{y}$. Then, for instance,
\[
g(10) = \mon{g}[\FF{x} \mapsto 1,\FF{y} \mapsto 0] = k_1+k_2.
\]
In a similar way we can express all $g(\vec{c})$ as sums of $k_i$. Thus, to find the $k_i$, we have to solve
\[
\left(\begin{array}{cccc}1&0&0&0\\1&1&0&0\\1&0&1&0\\1&1&1&1\end{array}\right)\left(\begin{array}{c}k_1\\k_2\\k_3\\k_4\end{array}\right) = \left(\begin{array}{c}3\\7\\-2\\4\end{array}\right).
\]
Since this matrix is lower triangular with nonzero diagonal entries, it is invertible, so the $k_i$ exist and are unique. In fact, we find $\mon{g} = 3+4\FF{x}-5\FF{y}+2\FF{x}\FF{y}$.
\end{example}

\section{The algorithm}

\begin{algorithm}[t]
\caption{The algorithm $\SQFU(T)$.} \label{alg:SQFU}
\SetKwInOut{Input}{input}\SetKwInOut{Output}{output}
\Input{A FT $T = (V,E,\gamma,\pp)$}
\Output{$U(T)$}
$\mathsf{ToDo} \leftarrow V$\;
\While{$\mathsf{ToDo} \neq \varnothing$}{
 Pick $v \in \mathsf{ToDo}$ minimal w.r.t. $\preceq$\;
 $\mathsf{ToDo} \leftarrow \mathsf{ToDo} \setminus \{v\}$\;
 \uIf{$\gamma(v) = \tBE$}{
  $g_v \leftarrow \pp(v)$\;
 }
 \Else{
  \uIf{$\gamma(v) = \tOR$}{
   $g_v \leftarrow 1-\prod_{w \in \ch(v)} (1-\FF{w})$\;
  }
  \Else{
   $g_v \leftarrow \prod_{w \in \ch(v)} \FF{w}$\;
  }
  $\mathsf{ToDo}_v \leftarrow \{w \in V \mid \idom(w) = v\}$\;
  \While{$\mathsf{ToDo}_v \neq \varnothing$}{
   Pick $w \in \mathsf{ToDo}_v$ maximal w.r.t. $\preceq$\;
   $\mathsf{ToDo}_v \leftarrow \mathsf{ToDo}_v \setminus \{w\}$\;
   $g_v \leftarrow g_v[\FF{w} \mapsto g_w]$\;
  }
 }
}
\Return{$g_{\RRR{T}}$}
\end{algorithm}

Using the notation of the previous two sections, we can now state our algorithm for calculating unreliability. It is presented in Algorithm \ref{alg:SQFU}. The algorithm works bottom-up, assigning to each node $v$ a formal expression $g_v \in \SFPA(\{w \in V \mid w \prec v\})$ representing the failure probability of $v$; the formal variables $\FF{w}$ present in $g_v$ represent nodes with multiple paths from the root, which we will also encounter further in the calculation.

The algorithm works as follows: working bottom-up (lines 1--4), the algorithm first assigns a $g_v$ of the most basic form to $v$ (lines 5--12): for a BE this is simply its failure probability $\pp(v)$, while for OR- and AND-gates it is the expression for the failure probability in terms of the formal variables $\FF{w}$, where $w$ ranges over $\ch(v)$. After obtaining this expression of $g_v$, the algorithm then substitutes away all formal variables that we will not encounter later in the computation (lines 13--18). These are precisely the $\FF{w}$ for which $\idom(w) = v$, as for these $w$ this is the point where we will not encounter other copies of $\FF{w}$ anymore. We replace each $\FF{w}$ with the associated expression $g_w$; we start with the $w$ closest to $v$, as these $g_w$ may contain other $\FF{w'}$ that also need to be substituted away. Finally, we return $g_{\RRR{T}}$ (line 21). At this point all formal variables have been substituted away, so $g_{\RRR{T}} \in \mathbb{R}$. Note that `under the hood' we have determined $\idom(v)$ for each $v$, which can be done in linear time by Theorem \ref{thm:idom}.

The main theoretical result of this paper is the validity of Alg.~\ref{alg:SQFU}.

\begin{theorem}\label{thm:main}
Let $T$ be a FT. Then $\SQFU(T) = U(T)$.
\end{theorem}

We will prove this theorem in Section \ref{sec:proof}. First, we introduce a slight extension to the FT formalism.

\section{Partially controllable fault trees}

In this section, we slightly extend the FT formalism in a manner necessary for the proof of Theorem \ref{thm:main}. The resulting objects, \emph{partially controllable fault trees} (PCFTs), are just like regular FTs, except that certain BEs are labelled \emph{controllable BEs}; these do not have a fixed failure probability but instead can be set to $\mathtt{0}$ or $\mathtt{1}$ at will. We emphasize that the concept of PCFTs does not correspond to an engineering reality, but is a mathematical construct needed for the proof of Theorem \ref{thm:main}.

\begin{definition}
An \emph{partially controllable fault tree} (PCFT) is a tuple $T = (V,E,\gamma,\pp)$ where:
\begin{itemize}
    \item $(V,E)$ is a rooted directed acyclic graph;
    \item $\gamma$ is a function $\gamma\colon V \rightarrow \{\tOR,\tAND,\tBE,\tCBE\}$ such that $\gamma(v) \in \{\tBE,\tCBE\}$ if and only if $v$ is a leaf;
    \item $\pp$ is a function $\pp\colon \BE{T} \rightarrow [0,1]$, where $\BE{T} = \{v \in V \mid \gamma(v) = \tBE\}$.
\end{itemize}
\end{definition}

Similar to $\BE{T}$ we define $\CBE{T} = \{v \in V \mid \gamma(v) = \tCBE\}$. Since the failure of CBEs is not probabilistic, one can only speak of the failure probability of $T$ once one has set the states of the CBEs. Therefore, $U(T)$ is not a fixed probability, but a function $\BB^{\CBE{T}} \rightarrow [0,1]$.

\begin{definition} \label{def:cbestruc}
\begin{enumerate}
\item The structure function of $T$ is a map $V \times \BB^{\BE{T}} \times \BB^{\CBE{T}} \rightarrow \BB$ defined by 
\[
\struc{T}(v,\vec{f},\vec{c}) = \begin{cases}
    f_v, & \textrm{ if $\gamma(v) = \tBE$,}\\
    c_v, & \textrm{ if $\gamma(v) = \tCBE$,}\\
    \bigvee_{w \in \ch(v)} \struc{T}(w,\vec{f},\vec{c}), & \textrm{ if $\gamma(v) = \tOR$},\\
    \bigwedge_{w \in \ch(v)} \struc{T}(w,\vec{f},\vec{c}), & \textrm{ if $\gamma(v) = \tAND$}.
\end{cases}
\]
\item Let $\vec{F}\in \BB^{\BE{T}}$ be a random variable so that $F_v$ is Bernoulli distributed with $\prob(F_v = \tone) = \pp(v)$ for each $v \in \BE{T}$, and all $F_v$ are independent. Then the \emph{unreliability} of $T$ is the function $U(T)\colon \BB^{\CBE{T}} \rightarrow [0,1]$ given by
\[
U(T)(\vec{c}) = \mathbb{P}(\struc{T}(\RRR{T},\vec{F},\vec{c}) = \tone).
\]
\end{enumerate}
\end{definition}

In light of Theorem \ref{thm:realfunc}, the function $U(T)\colon \BB^{\CBE{T}} \rightarrow \BB$ is described by its associated polynomial
\[
\langle U(T) \rangle \in \SFPA(\CBE{T}).
\]

\begin{figure}[b]
\centering
\includegraphics[width=2cm]{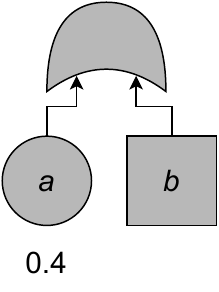}
\caption{The PCFT of Example \ref{ex:pcft}.} \label{fig:pcft}
\end{figure}

\begin{example} \label{ex:pcft}
Consider the PCFT $\tOR(a,b)$, where $\gamma(a) = \tBE$ and $\gamma(b) = \tCBE$ (see Fig.~\ref{fig:pcft}) and $\pp(a) = 0.4$. Then $\BB^{\tBE{T}} \cong \BB^{\tCBE{t}} \cong \BB$, and $\struc{T}(\RRR{T},f,c) = \tone$ if and only if at least one of $f_a,c_b$ equals $\tone$. Since $\prob(F_a = 1) = \pp(a) = 0.4$, it follows that
\[
U(T)(c) = \begin{cases}
0.4,& \textrm{ if $c_b = \tzero$},\\
1,& \textrm{ if $c_b = \tone$}.
\end{cases}
\]
As a polynomial this is $\langle U(T) \rangle = 0.4+0.6\FF{b}$.
\end{example}

\subsection{Quasimodular composition}

\begin{figure} \label{fig:mod}
\centering

\begin{subfigure}{0.2\textwidth}
\includegraphics[width=2.5cm]{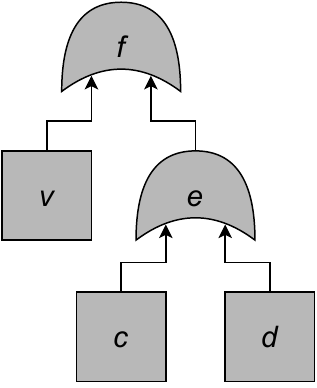}
\caption{$T$}
\end{subfigure}
\begin{subfigure}{0.25\textwidth}
\includegraphics[width=3cm]{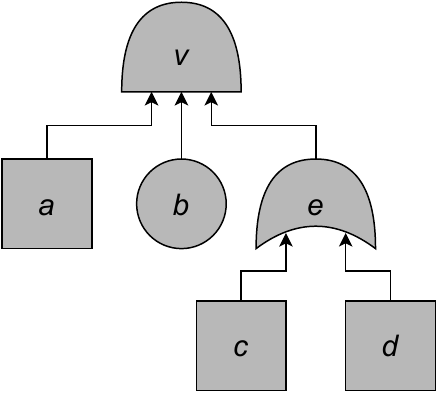}
\caption{$T'$}
\end{subfigure}
\begin{subfigure}{0.25\textwidth}
\includegraphics[width=3cm]{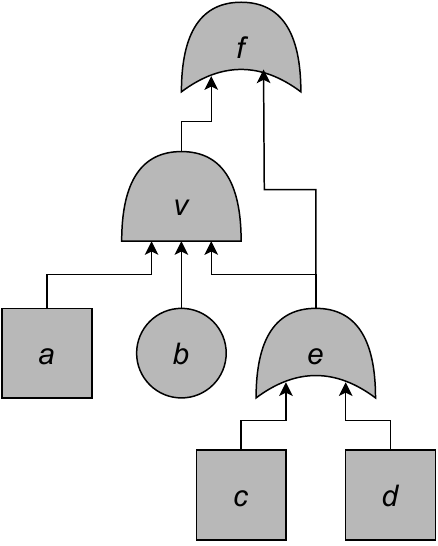}
\caption{$T[v \mapsto T']$}
\end{subfigure}
\caption{An example of quasimodular composition. Square nodes are CBEs.} \label{fig:mod}
\end{figure}

Now that we have expressed a PCFT $T$ as a polynomial $\langle U(T) \rangle$, the next step is to relate substitution operations on such polynomials to graph-theoretic operations on PCFTs. The key concept on the PCFT side is \emph{quasimodular composition}, which is defined as follows.

\begin{definition} \label{def:mod}
Let $T = (V,E,\gamma,\pp)$ and $T' = (V',E',\gamma',\pp')$ where $V$ and $V'$ are not necessarily disjoint, such that $E,\gamma,\pp,\ch$ coincide with $E',\gamma',\pp',\ch$ on $V \cap V'$. Let $v \in \CBE{T} \setminus V'$, and assume that $\BE{T} \cap \BE{T'} = \varnothing$ (see Fig.~\ref{fig:mod}). Then the \emph{quasimodular composition} $T[v \mapsto T']$ of $T$ and $T'$ in $v$ is the PCFT obtained by replacing $v$ in $T$ by the entire PCFT $T'$, rerouting all edges originally to $v$ to $\RRR{T'}$ instead.
\end{definition}

The concept of quasimodular composition of PCFTs is closely related to modular composition of FTs \cite{rauzy1997exact}. The difference is that in modular composition $T$ and $T'$ may not share any nodes, while in quasimodular composition they may share CBEs, as well as any internal nodes; however, due to the condition that these internal nodes must have the same children in $T$ and $T'$, any shared internal nodes may not have any BE descendants.

The following result states that substitution on the polynomial level precisely corresponds to quasimodular composition on the PCFT level; it is the key ingredient to the proof of Theorem \ref{thm:main}.

\begin{theorem} \label{lem:mod}
Let $T,T',v$ be as in Definition \ref{def:mod} and let $T'' = T[v \mapsto T']$ be their quasimodular composition. Then $\CBE{T''} = \CBE{T}\setminus\{v\} \cup \CBE{T'}$, and as elements of $\SFPA(\CBE{T''})$ one has
\[
\mon{U(T'')} = \mon{U(T)}[\FF{v} \mapsto \mon{U(T')}].
\]
\end{theorem}

This theorem is best read `in reverse': given a large PCFT $T''$, one can calculate $U(T'')$ by finding a \emph{quasimodule} $T'$ and its remainder $T$, and combining $U(T)$ and $U(T')$. In this sense, this theorem is analogous to \emph{modular decomposition} of FTs \cite{rauzy1997exact}, in which a FT's unreliability is expressed in terms of that of its \emph{modules}. Again, the key difference is that we allow not just modular decomposition, but also quasimodular decomposition.

\begin{figure}
\centering
\begin{subfigure}{0.14\textwidth}
\includegraphics[width=2cm]{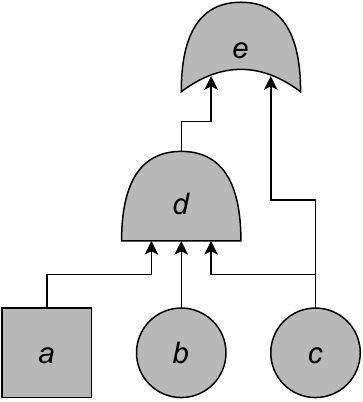}
\caption{$T$}
\end{subfigure} \quad 
\begin{subfigure}{0.14\textwidth}
\includegraphics[width=2cm]{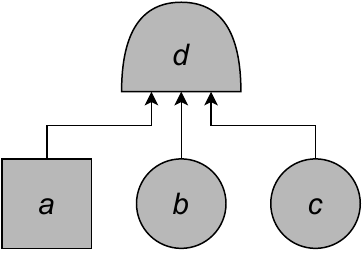}
\caption{$T_d$}
\end{subfigure} \quad 
\begin{subfigure}{0.14\textwidth}
\includegraphics[width=1.2cm]{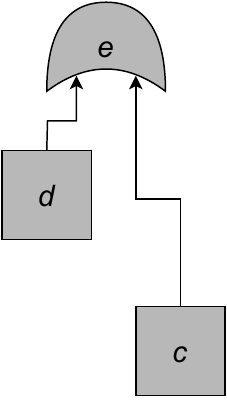}
\caption{$T[\{c,d\}]$}
\end{subfigure}
\caption{An example of the constructions of Definition \ref{def:subtree}. A PCFT $T$ is depicted in (a) (square nodes are CBEs). The sub-FT $T_d$ with root $d$ is depicted in (b). The FT $T[\{c,d\}]$ obtained by turning $c,d$ into CBEs is depicted in (c); note that this FT is also equal to $T[\{b,c,d\}]$, $T[\{a,c,d\}]$ and $T[\{a,b,c,d\}]$.} \label{fig:cons}
\end{figure}

\section{Sketched proof of correctness} \label{sec:proof}

In this section, we sketch the proof of Theorem \ref{thm:main}; a full proof is presented in the appendix. Before we outline the proof, we first define two ways to construct new (PC)FTs from a FT.

\begin{definition} \label{def:subtree}
Let $T = (V,E,\gamma,\pp)$ be a FT.
\begin{enumerate}
    \item Let $v \in V$. Then $T_v = (V_v,E_v,\gamma_v,\pp_v)$ is the FT consisting of the descendants of $v$, with $v$ as a root.
    \item Let $I \subseteq V$. Then $T[I]$ is the PCFT obtained from $T$ via the following procedure:
\begin{itemize}
\item For each $v \in I$, set $\gamma(v) = \tCBE$;
\item For each $v \in I$, remove all outgoing edges;
\item Then $T[I]$ is the PCFT consisting of all nodes reachable from the root.
\end{itemize}
\end{enumerate}
\end{definition}

These constructions are depicted in Figure \ref{fig:cons}.

For a node $v$, let $g_{v,\infty}$ be the value of $g_v$ at the end of the loop in lines 14--18 of Algorithm \ref{alg:SQFU}; this is the value $g_{v}$ has when the algorithm ends, and which will be used to substitute $\FF{v}$ in line 17. Then Theorem \ref{thm:main} follows from the following result:

\begin{theorem} \label{thm:mainext}
Let $v \in V$, and define 
\[
\mathcal{I}_v = \{w \in V \mid w \prec v \prec \idom(w)\}.
\]
Then $g_{v,\infty} = \mon{U(T_v[\mathcal{I}_v])}$.
\end{theorem}

Theorem \ref{thm:main} is just the special case $v = \RRR{T}$, as $T_{\RRR{T}} = T$ and $\mathcal{I}_{\RRR{T}} = \varnothing$. The proof can be sketched as follows:

This is proven by induction on $T$. For BEs this is immediate. If $\gamma(v) = \tAND$, then $g_v$ is initialized as $\prod_{w \in \ch(v)} \FF{w}$. Then, for each $w$ picked in line 15 of Algorithm \ref{alg:SQFU}, a $\FF{w}$ is substituted by its corresponding polynomial $g_w = g_{w,\infty}$. By the induction hypothesis, $g_{w,\infty}$ corresponds to the unreliability function of a certain PCFT, and by Theorem \ref{lem:mod}, this substitution operation corresponds to the composition of PCFTs. By keeping track of the form of the resulting PCFT, one shows that the PCFT one ends up with is exactly $T_v[\mathcal{I}_v]$, showing the result for $v$. The case $\gamma(v) = \tOR$ is, of course, completely analogous. By induction, this proves Theorem \ref{thm:mainext}, and by consequence Theorem \ref{thm:main}.

\section{Complexity} \label{sec:comp}

\begin{algorithm}[t]
\caption{The algorithm $\SQFUU(T)$.} \label{alg:SQFUU}
\SetKwInOut{Input}{input}\SetKwInOut{Output}{output}
\Input{A FT $T = (V,E,\gamma,\pp)$}
\Output{$U(T)$}
$\mathsf{ToDo} \leftarrow V$\;
\While{$\mathsf{ToDo} \neq \varnothing$}{
 Pick $v \in \mathsf{ToDo}$ minimal w.r.t. $\preceq$\;
 $\mathsf{ToDo} \leftarrow \mathsf{ToDo} \setminus \{v\}$\;
 \uIf{$\gamma(v) = \tBE$}{
  $g_v \leftarrow \pp(v)$\;
 }
 \Else{
 $S_v = \{w \in \ch(v) \mid \idom(w) = v\}$\;
  \uIf{$\gamma(v) = \tOR$}{
  $p_1 \leftarrow \prod_{w \in S_v} (1-g_w)$\;
  $p_2 \leftarrow \prod_{w \in \ch(v) \setminus S_v} (1-\FF{w})$\;
  $g_v \leftarrow 1-p_1p_2$\;
  }
  \Else{
  $p_1 \leftarrow \prod_{w \in S_v} g_w$\;
  $p_2 \leftarrow \prod_{w \in \ch(v) \setminus S_v} \FF{w}$\;
  $g_v \leftarrow p_1p_2$\;
  }
  $\mathsf{ToDo}_v \leftarrow \{w \in V\setminus\ch(v) \mid \idom(w) = v\}$\;
  \While{$\mathsf{ToDo}_v \neq \varnothing$}{
   Pick $w \in \mathsf{ToDo}_v$ maximal w.r.t. $\preceq$\;
   $\mathsf{ToDo}_v \leftarrow \mathsf{ToDo}_v \setminus \{w\}$\;
   $g_v \leftarrow g_v[\FF{w} \mapsto g_w]$\;
  }
 }
}
\Return{$g_{\RRR{T}}$}
\end{algorithm}

The complexity of Alg.~\ref{alg:SQFU} can be bound in terms of graph parameters of the DAG $T$. To do so, we first note that we can slightly rephrase the algorithm as follows. If $w$ is a child of $v$ and $w$ has only one parent, then $w$ is a maximal element of $\mathsf{ToDo}_v$. As such, in line 15 these $w$ will be picked first. Therefore, we may as well do this replacement in lines 9 and 11 directly. This leads to Alg.~\ref{alg:SQFUU}, which has the same functionality as Alg.~\ref{alg:SQFU} and therefore calculates $U(T)$ correctly. Note that the condition that $w$ has only one parent is equivalent to $\idom(w) = v$, which leads to our definition of $S_v$ in line 8.

Like the standard algorithm for reliability analysis for treelike FTs \cite{ruijters2015fault}, Alg.~\ref{alg:SQFUU} works bottom-up. However, it is more complicated due to the fact that our main objects of interest are squarefree polynomials rather than real numbers, and operating on these induces a larger complexity. To describe this complexity, we introduce the following notation:
\[
X = \{v \in V \mid v \textrm{ has multiple parents}\}.
\]
Since in Alg.~\ref{alg:SQFUU} line 8, the set $S_v$ consists of all children of $v$ with a single parent, the only $\FF{x}$ that are introduced satisfy $x \in X$. Thus each $g_v$ is an element of $\SFPA(X)$, and as such has at most $2^{|X|}$ terms. Multiplying two such polynomials has complexity $\mathcal{O}(4^{|X|})$, and since substitution is just multiplication by Lemma \ref{lem:sub}, substitution has the same complexity. Next, we count the number of multiplications and substitutions. The FT $T$ has $|X|$ nodes with more than 1 parent and 1 node with $0$ parents, so in total $|V|-|X|-1$ nodes have exactly one parent and are used in multiplications in lines 10 and 14 of Alg.~\ref{alg:SQFUU}; in particular, there are at most $|V|$ multiplications. Furthermore, for each $v$ one has $|\mathsf{ToDo}_v| \leq |X|$, hence at each $v$ there are at most $|X|$ substitutions in line 22. We conclude:

\begin{theorem} \label{thm:comp1}
Let $X$ be as above. Then Alg.~\ref{alg:SQFUU} has time complexity $\mathcal{O}(|V|(|X|+1)4^{|X|})$.
\end{theorem}

If $X = \varnothing$, then $T$ is treelike. In this case, Theorem \ref{thm:comp1} tells us that Alg.~\ref{alg:SQFUU} has time complexity $\mathcal{O}(|V|)$. Indeed, in this case $S_v = \ch(v)$ for all $v$, and each $g_v$ is a real number. Hence Alg.~\ref{alg:SQFUU} reduces to the standard bottom-up algorithm for treelike FTs, which is known to have linear time complexity. In fact, Theorem \ref{thm:comp1} generalizes this result: it shows that for bounded $|X|$, time complexity of Alg.~\ref{alg:SQFUU} is linear. This makes Alg.~\ref{alg:SQFUU} a useful tool if the non-tree topology of an FT is concentrated in only a few nodes.

The complexity bound of Theorem \ref{thm:comp1} can be sharpened by realizing that a variable $\FF{w}$ only occurs in the computation of $g_v$ if $w \prec v \preceq \idom(w)$. Using this as a bound we get the following result:

\begin{theorem} \label{thm:comp2}
Define
\[
c = \max_{v \in V} |\{w \in X \mid w \prec v \preceq \idom(w)\}|.
\]
Then Alg.~\ref{alg:SQFUU} has time complexity $\mathcal{O}(|V|(c+1)4^c)$.
\end{theorem}

As far as I am aware, a comparable analysis that bounds the time complexity of BDD-based methods in terms of multiparent nodes does not exist in the literature. Such an analysis can be complicated by variable ordering heuristics for the BDD construction, and is beyond the scope of this paper. Thus the existence of a provable complexity bound in terms of the number of multiparent nodes is a major advantage of the SFPA method.

\section{Experiments}

\begin{table} \label{tab:results}
\centering
\begin{tabular}{llcc}
 & & SFPA & Storm \\ \hline 
\multirow{3}{*}{Benchmark 1} & minimum & 0.60 &  0.89  \\
 & median & 3.51  & 1.81 \\
 & maximum & 60${}^*$ &  7.91 \\ \hline
 \multirow{3}{*}{Benchmark 2} & minimum & 0.30  & 0.41 \\
 & median & 0.68  & 1.42 \\
 & maximum & 6.73  & 47.29 
\end{tabular}
\caption{Summary of the results: the minimum, median and maximum of the computation times (in seconds) of the two algorithms are displayed. ${}^*$Timeout at 60 seconds, attained 3 times.} \label{tab:results}
\end{table}


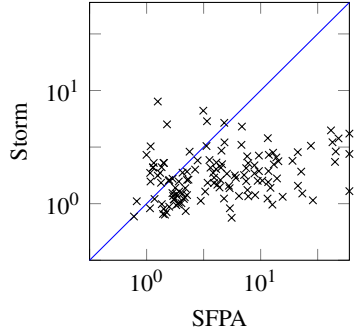
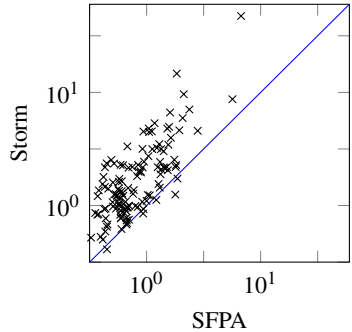
\begin{figure} \label{fig:results}
\centering
\begin{subfigure}{0.8\linewidth}
    		\begin{tikzpicture}
		\tikzstyle{every node}=[font=\small]
		\begin{axis}[
		width = 5cm,
		height = 5cm,
		xmin = -0.5,
		xmax = 1.7782,
		ymin = -0.5,
		ymax = 1.7782,
            xlabel = SFPA,
            ylabel = Storm,
		ylabel near ticks,
		xtick = {-0.5,0,0.5,1,1.5,2},
		xticklabels = {,$10^0$,,$10^1$,,$10^2$},
		ytick = {-0.5,0,0.5,1,1.5,2},
		yticklabels = {,$10^0$,,$10^1$,,$10^2$},
		]
		\addplot[only marks,mark=x] table[col sep = comma] {random_scram.csv};
		\addplot[mark=none, color=blue] coordinates {(-0.5,-0.5) (1.7782,1.7782)};
		\end{axis}
		\end{tikzpicture}
  \caption{Unreliability on benchmark set 1.}
  \vspace{1em}
\end{subfigure}
\begin{subfigure}{0.8\linewidth}
    		\begin{tikzpicture}
		\tikzstyle{every node}=[font=\small]
		\begin{axis}[
		width = 5cm,
		height = 5cm,
		xmin = -0.5,
		xmax = 1.7782,
		ymin = -0.5,
		ymax = 1.7782,
            xlabel = SFPA,
            ylabel = Storm,
		ylabel near ticks,
		xtick = {-0.5,0,0.5,1,1.5,2},
		xticklabels = {,$10^0$,,$10^1$,,$10^2$},
		ytick = {-0.5,0,0.5,1,1.5,2},
		yticklabels = {,$10^0$,,$10^1$,,$10^2$},
		]
		\addplot[only marks,mark=x] table[col sep = comma] {almost_trees.csv};
		\addplot[mark=none, color=blue] coordinates {(-0.5,-0.5) (1.7782,1.7782)};
		\end{axis}
		\end{tikzpicture}
  \caption{Unreliability on benchmark set 2.}
\end{subfigure}
\caption{Timing comparison of SFPA and Storm for calculating unreliability. Times are in seconds, timeout at 60s.} \label{fig:results}
\end{figure}

We perform experiments to test $\SQFU$'s performance, as implemented in Python. All experiments are performed on a Ubuntu virtual machine with 6 cores and 8GB ram, running on a PC with an Intel  Core i7-10750HQ 2.8GHz processor and 16GB memory. We compare the performance to that of Storm-dft, a state-of-the-art model checker that calculates FT unreliability via a BDD-based approach with modularization \cite{basgoze2022bdds}. We compare performance on two benchmark sets of FTs:
\begin{enumerate}
\item A collection of 128 randomly generated FTs used as a benchmark set in \cite{basgoze_daniel_2022_6390998}. These FTs have, on average, 89.3 nodes, of which 10.3 have multiple parents.
\item A new randomly generated collection of 128 FTs, created using SCRAM \cite{SCRAM}. These FTs have, on average, 123.8 nodes, 4.1 of which have multiple parents.
\end{enumerate}

The second benchmark set was created in order to validate the theoretical results of Section \ref{sec:comp}, where it was shown that the complexity of SFPA depends on the number of nodes with multiple parents. We compute both the unreliability of each FT and measure the time of both computations (timeout: 60 seconds). The results are given in Table \ref{tab:results} and Figure \ref{fig:results}. As one can see, Storm largely outperforms SFPA on benchmark set 1, with lower computation times on 76\% of the unreliability calculations.
On benchmark set 2, SFPA fares considerably better, outperforming Storm on 95\% of all FTs for unreliability calculation: on average SFPA takes only 54\% of the computation time of Storm. The fact that SFPA is more efficient on this benchmark set can be understood from Theorem \ref{thm:comp1}, which shows that computational complexity of SFPA is low when the number of multiparent nodes is low. By contrast, for a BDD-based approach the presence of \emph{any} multiparent nodes means one cannot use the bottom-up algorithm and has to rely on creating the BDD, which is usually slower than the bottom-up approach. Furthermore, modularization may only be of limited use depending on the position of the multiparent nodes.


Overall, we can conclude that for calculating unreliability SFPA is competitive with the state-of-the-art, and is significantly faster on an FT benchmark set with fewer multiparent nodes.

\section{Conclusion}

In this paper, we have introduced SFPA, a novel algorithm for calculating fault tree unreliability based on squarefree polynomial algebras. We have proven its validity and given complexity bounds in terms of the number of multiparent nodes. Experiments show that it is significantly faster than the state of the art on FTs with few multiparent nodes.

There are several directions for future work. First, our proof-of-concept Python implementation of SFPA can undoubtedly be improved, leading to faster computation. Such improvements can be done on the theoretical side as well. For example, one could introduce a new formal variable $\mathsf{U}_v$ for $1-\FF{v}$; this would decrease the number of terms in the expression of $g_v$ when $v$ is an $\tOR$-gate from $2^{|\ch(v)|}-1$ to $2$, hopefully leading to faster computation. In this case, new computation rules such as $\FF{v}\mathsf{U}_v = 1$ need to be introduced.

Second, our experimental results show that a BDD-based method works best for FTs with more multiparent nodes, while SFPA works best for FTs with fewer multiparent nodes. It would be interesting to see a more extensive experimental evaluation that investigates what the break-even point is. Such an experimental evaluation can be augmented by incorporating real-world case studies, to test the effectiveness of SFPA in practice.

On the other hand, it would be interesting to see to what extent SFPA-like methods can be applied to other problems in FT analysis, such as the analysis of dynamic FTs, which also consider time-dependent gates and behaviour. A good candidate is the analysis of attack trees (ATs), the security counterpart of FTs. Quantitative analysis of (non-dynamic) ATs is also done using BDDs \cite{lopuhaa2022efficient}, which has the same issues as BDD-based FT analysis. We expect that SFPA-like methods can be extended to ATs as well.

\section*{Acknowledgements}

This research has been partially funded by ERC Consolidator grant 864075 CAESAR and the European Union’s Horizon 2020 research and innovation programme under the Marie Skłodowska-Curie grant agreement No. 101008233.

\bibliographystyle{apalike}
{\small
\bibliography{mybibliography}}


\section*{\uppercase{Appendix}}

\subsection*{Proof of Theorem \ref{thm:NP}}

Before we prove the theorem itself, we first introduce some notation. Consider an unaugmented FT $T' = (V,E,\gamma)$, i.e., a FT without $\pp$ specified. This does not affect the definition of the structure function or cut sets. For $\vec{f},\vec{f}' \in \BB^{\BE{T'}}$, we write $\vec{f} \preceq \vec{f}'$ if $f_v \leq f'_v$ for all $v \in \BE{T}$. A cut set $\vec{f}$ is called \emph{minimal} if there does not exist another cut set $\vec{f}'$ with $\vec{f}' \prec \vec{f}$. The following problem is NP-hard \cite{rauzy1993new}:

\begin{problem} \label{prob:mcs}
Given an unaugmented FT $T'$, find one of its minimal cut sets.
\end{problem}

Therefore, Theorem \ref{thm:NP} is a consequence of the following result:

\begin{theorem}
Problem \ref{prob:mcs} can be reduced to Problem \ref{prob:ut}.
\end{theorem}

\begin{proof}
Let $T' = (V,E,\gamma)$ be an unaugmented FT, and take any enumeration $\BE{T'} = \{v_0,\ldots,v_{n-1}\}$. Define $\pp\colon \BE{T'} \rightarrow [0,1]$ by $\pp(v_i) = 10^{-2^i}$, and let $T = (V,E,\gamma,\pp)$ be the resulting augmented fault tree. Let $M \subseteq \CS{T}$ be the set of minimal cut sets. For $\vec{f} \in M$, define $\kappa(\vec{f}) = -\log_{10}\left(\prod_{f_b = \tone} \pp(b)\right)$; let $\vec{g} \in M$ be such that $\kappa(\vec{g})$ is minimal.

Let $\vec{F}$ be the random variable in $\BB^{\BE{T}}$ defined in Definition \ref{def:unreliability}; then $\prob(\vec{f} \preceq \vec{F}) = 10^{-\kappa(\vec{f})}$ for all $\vec{f} \in M$. Since $\vec{F} \in \CS{T}$ if and only if there is a $\vec{f} \in M$ such that $\vec{f} \preceq \vec{F}$, it follows that

\begin{align}
10^{-\kappa(\vec{g})} &= \prob(\vec{g} \preceq \vec{F}) \nonumber \\
&\leq \prob(\exists \vec{f} \in M\colon \vec{f} \preceq \vec{F}) \nonumber\\
&= U(T) \label{eq:comp1}\\
&\leq \sum_{\vec{f} \in M} \prob(\vec{f} \preceq \vec{F}) \nonumber \\
&= \sum_{\vec{f} \in M} 10^{-\kappa(\vec{f})}. \label{eq:comp2}
\end{align}

We now give an upper bound on $\sum_{\vec{f} \in M} 10^{-\kappa(\vec{f})}$ in terms of $\kappa(\vec{g})$. If $\{v \in \BE{T} \mid f_v = \tone\} = \{v_{i_1},\ldots,v_{i_k}\}$, then $\kappa(\vec{f}) = \sum_{j=1}^k 2^{i_j}$. Thus, each $\kappa(\vec{f})$ is an integer, whose binary representation is equal to $\vec{f}$ (i.e., its digit for $2^i$ indicates whether $f_{v_i} = \tone$). As a result, the set $\{\kappa(\vec{f}) \mid \vec{f} \in M\}$ is a set of distinct integers, the least of which is $\kappa(\vec{g})$. Hence
\begin{equation}
\sum_{\vec{f} \in M} 10^{-\kappa(\vec{f})} \leq \sum_{i=0}^{|M|-1} 10^{-\kappa(\vec{g})-i} < 10^{-\kappa(\vec{g})+1} \label{eq:comp3}
\end{equation}
From \eqref{eq:comp1}, \eqref{eq:comp2}, \eqref{eq:comp3} it follows that that $-\lfloor \log_{10}(U(T)) \rfloor = \kappa(\vec{g})$. Thus from $U(T)$ we find $\kappa(\vec{g})$ in polynomial time, and since $\vec{g}$ is the binary representation of the integer $\kappa(\vec{g})$, we also find the MCS $\vec{g}$.
\end{proof}

\subsection*{Proof of Lemma \ref{lem:idom}}

Suppose $\idom(v) \not \preceq w$. Since $v \prec w$, any path $\RRR{T} \rightarrow w$ can be extended to a path $\RRR{T} \rightarrow v$; hence $\idom(v)$ lies on the extended path. Since $\idom(v) \not \preceq w$, it cannot lie on the subpath $w \rightarrow v$, so it lies on the path $\RRR{T} \rightarrow w$. Since this is true for every path $\RRR{T} \rightarrow w$, it follows that $\idom(v)$ is a dominator of $w$; hence $\idom(w) \preceq \idom(v)$. \qed

\subsection*{Proof of Lemma \ref{lem:sub}}

We have
\begin{align*}
&\alpha[\FF{x} \mapsto \beta]\\
&= \sum_{\substack{Y \subseteq X\colon\\x \in Y}} \alpha_Y \left(\prod_{x' \in Y\setminus\{x\}} \FF{x'}\right)\cdot \beta + \sum_{\substack{Y \subseteq X\colon\\x \notin Y}} \alpha_Y \left(\prod_{x' \in Y} \FF{x'}\right) \\
&= \left(\sum_{\substack{Y \subseteq X\colon\\x \in Y}} \alpha_Y \prod_{x' \in Y\setminus\{x\}} \FF{x'}+ \sum_{\substack{Y \subseteq X\colon\\x \notin Y}} \alpha_Y \prod_{x' \in Y} \FF{x'}\right)\cdot \beta\\
& \quad \quad + \left(\sum_{\substack{Y \subseteq X\colon\\x \notin Y}} \alpha_Y \prod_{x' \in Y} \FF{x'}\right)\cdot (1-\beta)\\
&= \left(\sum_{Y \subseteq X} \alpha_Y \prod_{x' \in Y\setminus\{x\}} \FF{x'}\right)\cdot \beta\\
& \quad \quad + \left(\sum_{\substack{Y \subseteq X\colon\\x \notin Y}} \alpha_Y \prod_{x' \in Y} \FF{x'}\right)\cdot (1-\beta)\\
&= \alpha[\FF{x} \mapsto 1]\cdot \beta + \alpha[\FF{x} \mapsto 0] \cdot (1-\beta). \qed
\end{align*}

\subsection*{Proof of Lemma \ref{lem:subarit}}

The first statement follows immediately from the coefficientwise definitions of addition and substitution. For the second statement, we prove it first for the special cases $\beta=1$ and $\beta=0$. In the first case, for each $Y \subseteq X \setminus \{x\}$ one has
\[
\alpha[\FF{x} \mapsto 1]_Y = \alpha_Y + \alpha_{Y \cup \{x\}}.
\]
It follows that
\begin{align*}
&(\alpha_1\alpha_2)[\FF{x} \mapsto 1]_Y\\
&= (\alpha_1\alpha_2)_Y + (\alpha_1\alpha_2)_{Y \cup \{x\}} \\
&= \sum_{\substack{Y_1,Y_2 \subseteq Y\colon\\ Y_1 \cup Y_2 = Y}} \alpha_{1,Y_1}\alpha_{2,Y_2} + \sum_{\substack{Y'_1,Y'_2 \subseteq Y\cup \{x\}\colon\\ Y'_1 \cup Y'_2 = Y\cup \{x\}}} \alpha_{1,Y'_1}\alpha_{2,Y'_2}.
\end{align*}
In the second summation we can distinguish the cases $x \in Y'_1 \setminus Y'_2$, $x \in Y'_2 \setminus Y'_1$ and $x \in Y'_1 \cap Y'_2$. Substituting $Y_i := Y'_i \setminus \{x\}$ in the second summation, we get 
\begin{align*}
&(\alpha_1\alpha_2)[\FF{x} \mapsto 1]_Y \\
&= \sum_{\substack{Y_1,Y_2 \subseteq Y\colon\\ Y_1 \cup Y_2 = Y}} \alpha_{1,Y_1}\alpha_{2,Y_2} 
 + \sum_{\substack{Y_1,Y_2 \subseteq Y\colon\\ Y_1 \cup Y_2 = Y}} \Big( \alpha_{1,Y_1\cup\{x\}}\alpha_{2,Y_2} \\
 & \quad \quad +\alpha_{1,Y_1}\alpha_{2,Y_2\cup\{x\}}+\alpha_{1,Y_1\cup\{x\}}\alpha_{2,Y_2\cup\{x\}}\Big)\\
&= \sum_{\substack{Y_1,Y_2 \subseteq Y\colon\\ Y_1 \cup Y_2 = Y}} \Big( \alpha_{1,Y_1}\alpha_{2,Y_2} + \alpha_{1,Y_1\cup\{x\}}\alpha_{2,Y_2}\\
&\quad \quad +\alpha_{1,Y_1}\alpha_{2,Y_2\cup\{x\}}+\alpha_{1,Y_1\cup\{x\}}\alpha_{2,Y_2\cup\{x\}}\Big)\\
&= \sum_{\substack{Y_1,Y_2 \subseteq Y\colon\\ Y_1 \cup Y_2 = Y}} \alpha_1[\FF{x} \mapsto 1]_{Y_1}\alpha_2[\FF{x} \mapsto 1]_{Y_2}\\
&= (\alpha_1[\FF{x} \mapsto 1]\cdot \alpha_2[\FF{x} \mapsto 1])_Y.
\end{align*}
This proves the statement for $\beta = 1$. The proof for $\beta = 0$ is analogous, this time using $\alpha[\FF{x} \mapsto 0]_Y = \alpha_Y$ for $Y \subseteq X \setminus \{x\}$.

Now we go to general $\beta$, for which we have
\begin{align}
&\alpha_1[\FF{x} \mapsto \beta] \cdot \alpha_2[\FF{x} \mapsto \beta] \nonumber\\
&= (\alpha_1[\FF{x} \mapsto 1]\cdot \beta + \alpha_1[\FF{x} \mapsto 0] \cdot (1-\beta)) \label{eq:subarit1} \\
&\quad \quad \cdot (\alpha_2[\FF{x} \mapsto 1]\cdot \beta + \alpha_2[\FF{x} \mapsto 0] \cdot (1-\beta)) \nonumber\\
&= \alpha_1[\FF{x} \mapsto 1]\cdot \alpha_2[\FF{x} \mapsto 1] \cdot \beta^2 \nonumber \\
&\quad \quad + \alpha_1[\FF{x} \mapsto 0]\cdot \alpha_2[\FF{x} \mapsto 0] \cdot (1-2\beta+\beta^2) \nonumber\\
&\quad \quad + \Big(\alpha_1[\FF{x} \mapsto 1]\cdot \alpha_2[\FF{x} \mapsto 0]\nonumber\\
&\quad \quad +\alpha_1[\FF{x} \mapsto 0]\cdot \alpha_2[\FF{x} \mapsto 1]\Big) \cdot (\beta-\beta^2) \nonumber \\
&= \alpha_1[\FF{x} \mapsto 1]\cdot \alpha_2[\FF{x} \mapsto 1] \cdot \beta \label{eq:subarit2}\\
&\quad \quad + \alpha_1[\FF{x} \mapsto 0]\cdot \alpha_2[\FF{x} \mapsto 0] \cdot (1-\beta) \nonumber \\
&= (\alpha_1\alpha_2)[\FF{x} \mapsto 1] \cdot \beta \label{eq:subarit3}\\
& \quad \quad + (\alpha_1\alpha_2)[\FF{x} \mapsto 0] \cdot (1-\beta) \nonumber\\
&= (\alpha_1\alpha_2)[\FF{x} \mapsto \beta]. \label{eq:subarit4}
\end{align}
Here we used Lemma \ref{lem:sub} in \eqref{eq:subarit1} and \eqref{eq:subarit4}, the fact that $\beta^2 = \beta$ in \eqref{eq:subarit2}, and our result for $\beta = 0$ and $\beta = 1$ in \eqref{eq:subarit3}. \qed

\subsection*{Proof of Lemma \ref{lem:subexchange}}

As for Lemma \ref{lem:subarit} we first prove this for $\beta_1,\beta_2 \in \{0,1\}$. First suppose $\beta_1 = \beta_2 = 1$. Using the rule $\alpha[\FF{x} \mapsto 1]_Y = \alpha_Y + \alpha_{Y \cup \{x\}}$ from the proof of Lemma \ref{lem:subarit}, we find, for $Y \subseteq X \setminus \{x_1,x_2\}$:
\begin{align*}
&\alpha[\FF{x_1} \mapsto 1][\FF{x_2} \mapsto 1]_Y\\
&= \alpha[\FF{x_1} \mapsto 1]_Y + \alpha[\FF{x_1} \mapsto 1]_{Y \cup \{x_2\}} \\
&= \alpha_Y + \alpha_{Y \cup \{x_1\}} + \alpha_{Y \cup \{x_2\}} + \alpha_{Y \cup \{x_1,x_2\}}\\
&= \alpha[\FF{x_2} \mapsto 1][\FF{x_1} \mapsto 1]_Y.
\end{align*}
The other cases for $\beta_1,\beta_2 \in \{0,1\}$ are completely analogous, also using the rule $\alpha[\FF{x} \mapsto 0]_Y = \alpha_Y$. Since the substitution order does not matter, we write e.g. $\alpha[\FF{x_1} \mapsto 1, \FF{x_2} \mapsto 0]$.

Now consider general $\beta_1,\beta_2$. Through repeated use of Lemmas \ref{lem:sub} and \ref{lem:subarit} we find
\begin{align*}
&\alpha[\FF{x_1} \mapsto \beta_1][\FF{x_2} \mapsto \beta_2]\\
&= \alpha[\FF{x_1} \mapsto \beta_1][\FF{x_2} \mapsto 1]\cdot \beta_2\\
&\quad + \alpha[\FF{x_1} \mapsto \beta_1][\FF{x_2} \mapsto 0] \cdot (1-\beta_2) \\
&= \Big(\alpha[\FF{x_1} \mapsto 1] \cdot \beta_1 + \alpha[\FF{x_1} \mapsto 0] \cdot (1-\beta_1)\Big)[\FF{x_2} \mapsto 1]\cdot \beta_2\\
&\quad + \Big(\alpha[\FF{x_1} \mapsto 1] \cdot \beta_1 \\
&\quad + \alpha[\FF{x_1} \mapsto 0] \cdot (1-\beta_1)\Big)[\FF{x_2} \mapsto 0]\cdot (1-\beta_2)\\
&= \Big(\left(\alpha[\FF{x_1} \mapsto 1] \cdot \beta_1\right)[\FF{x_2} \mapsto 1]\\
&\quad +\left(\alpha[\FF{x_1} \mapsto 0] \cdot (1-\beta_1)\right)[\FF{x_2} \mapsto 1] \Big)\cdot \beta_2\\
& \quad + \Big(\left(\alpha[\FF{x_1} \mapsto 1] \cdot \beta_1\right)[\FF{x_2} \mapsto 0]\\
&\quad +\left(\alpha[\FF{x_1} \mapsto 0] \cdot (1-\beta_1)\right)[\FF{x_2} \mapsto 0] \Big)\cdot (1-\beta_2)\\
&= \Big(\alpha[\FF{x_1} \mapsto 1][\FF{x_2} \mapsto 1] \cdot \beta_1[\FF{x_2} \mapsto 1]\\
&\quad +\alpha[\FF{x_1} \mapsto 0][\FF{x_2} \mapsto 1] \cdot (1-\beta_1)[\FF{x_2} \mapsto 1] \Big)\cdot \beta_2\\
& \quad + \Big(\alpha[\FF{x_1} \mapsto 1][\FF{x_2} \mapsto 0] \cdot \beta_1[\FF{x_2} \mapsto 0]\\
&\quad + \alpha[\FF{x_1} \mapsto 0][\FF{x_2} \mapsto 0] \cdot (1-\beta_1)[\FF{x_2} \mapsto 0] \Big)\cdot (1-\beta_2)\\
\end{align*}
By assumption, $x_2$ does not appear in $\beta_1$, so
\[\beta_1[\FF{x_2} \mapsto 1] = \beta_1[\FF{x_2} \mapsto 0] = \beta_1.\]
So this last expression is equal to 
\begin{align*}
&\Big(\alpha[\FF{x_1} \mapsto 1,\FF{x_2} \mapsto 1] \cdot \beta_1\\
&\quad +\alpha[\FF{x_1} \mapsto 0,\FF{x_2} \mapsto 1] \cdot (1-\beta_1) \Big)\cdot \beta_2\\
& \quad + \Big(\alpha[\FF{x_1} \mapsto 1,\FF{x_2} \mapsto 0] \cdot \beta_1\\
& \quad + \alpha[\FF{x_1} \mapsto 0,\FF{x_2} \mapsto 0] \cdot (1-\beta_1) \Big)\cdot (1-\beta_2)\\
&= \alpha[\FF{x_1} \mapsto 1,\FF{x_2} \mapsto 1] \cdot \beta_1\beta_2\\
& \quad +\alpha[\FF{x_1} \mapsto 0,\FF{x_2} \mapsto 1] \cdot (1-\beta_1)\beta_2\\
& \quad + \alpha[\FF{x_1} \mapsto 1,\FF{x_2} \mapsto 0] \cdot \beta_1(1-\beta_2)\\
&\quad + \alpha[\FF{x_1} \mapsto 0,\FF{x_2} \mapsto 0] \cdot (1-\beta_1)(1-\beta_2).
\end{align*}
Expanding $\alpha[\FF{x_2} \mapsto \beta_2][\FF{x_1} \mapsto \beta_1]$ in the same way yields the exact same result, so we conclude that $\alpha[\FF{x_1} \mapsto \beta_1][\FF{x_2} \mapsto \beta_2] = \alpha[\FF{x_2} \mapsto \beta_2][\FF{x_1} \mapsto \beta_1]$. \qed

\subsection*{Proof of Theorem \ref{thm:realfunc}}

Let $\recht{Map}(\BB^X,\RR)$ be the set of functions $\BB^X \rightarrow \RR$. Consider the map $\varrho\colon \SFPA(X) \rightarrow \recht{Map}(\BB^X,\RR)$ given by
\[
\varrho(g)(c) = g[\forall x \in X\colon \FF{x} \mapsto c_x]
\]
for all $g \in \SFPA(X)$ and $c \in \BB^X$. We will show that this map is bijective; this completes the proof of the theorem, as its inverse is then the sought map $f \mapsto \mon{f}$. 

Let $\preceq$ be the partial order on $\BB^X$ given by $c \preceq c'$ if and only if $c_x \leq c'_x$ for all $x \in X$. Furthermore, let $\BB^X = \{c^1,\ldots,c^{M}\}$ be an enumeration of $\BB^X$ (so $M = 2^{|X|}$) such that $c_i \preceq c_j$ implies $i \leq j$ for all $i,j \leq M$. Then $\recht{Map}(\BB^X,\RR)$ can be identified with $\RR^M$, sending $f$ to the vector $(f(c_1),\ldots,f(c_M))^{\intercal} \in \RR^M$. We can also identify $\SFPA(X)$ with $\RR^M$, by identifying $g$ with $(g_{X_1},\ldots,g_{X_M})$; here $X^i = \{x \in X \mid c^i_X = \tone\}$. By construction, one has $X_i \subseteq X_j$ if and only if $c^i \preceq c^j$.

Under these identifications, we can regard $\varrho$ as a map $\RR^M \rightarrow \RR^M$; we now investigate what form this map takes. Note that
\begin{align*}
\varrho(g)(c^i) &= \sum_{X' \subseteq X} g_{X'}\left(\prod_{x \in X'} c^i_x\right).
\end{align*}
The latter product equals $1$ if $X' \subseteq X_i$, and $0$ otherwise; hence
\begin{align}
\varrho(g)(c^i) &= \sum_{j\colon X_j \subseteq X_i} g_{X_j}. \label{eq:uptri}
\end{align}
Thus $\varrho$ is a linear map $\RR^M \rightarrow \RR^M$ and can be represented by a matrix $E$. By \eqref{eq:uptri}, one has
\[
E_{i,j} = \begin{cases}
1, \textrm{ if $c_i \preceq c_j$},\\
0, \textrm{ otherwise.}
\end{cases}
\]
In particular, $E$ is lower triangular and the diagonal entries are all $1$. It follows that $E$ is invertible, hence $\varrho$ is bijective.

\subsection*{Proof of Lemma \ref{lem:mod}}

Fix a $\vec{c}'' \in \BB^{\CBE{T''}}$, and let $\tilde{v}$ be the node in $T''$ that replaced $v$ (i.e., that has the function of $\RRR{T'}$). To prove the theorem, we need to show that
\begin{align}
&\mon{U(T)}[\FF{v} \mapsto \mon{U(T')}][\forall x \in \CBE{T''}\colon \FF{x} \mapsto c''_x]\nonumber\\
&= U(T'')(\vec{c}''). \label{eq:modproof1}
\end{align}
To prove this, note that we have
\begin{align*}
&\mon{U_T}[\FF{v} \mapsto \mon{U(T')}][\forall x \in \CBE{T''}\colon \FF{x} \mapsto c''_x]\\
&= \Big(\mon{U_T}[\FF{v} \mapsto 1]\cdot \mon{U_{T'}}\\
&\quad +\mon{U_{T}}[\FF{v} \mapsto 0]\cdot (1-\mon{U_{T'}})\Big)[\forall x \in \CBE{T''}\colon \FF{x} \mapsto c''_x]\\
&= \mon{U_T}[\FF{v} \mapsto 1,\forall x \in \CBE{T''}\colon \FF{x} \mapsto c''_x] \\
&\quad \quad \cdot \mon{U_{T'}}[\forall x \in \CBE{T''}\colon \FF{x} \mapsto c''_x]\\
&\quad + \mon{U_T}[\FF{v} \mapsto 0,\forall x \in \CBE{T''}\colon \FF{x} \mapsto c''_x] \\
&\quad \quad \cdot (1-\mon{U_{T'}}[\forall x \in \CBE{T''}\colon \FF{x} \mapsto c''_x]).
\end{align*}
Here we used Lemmas \ref{lem:sub} and \ref{lem:subarit}, as well as the fact that $(c''_x)^2 = c''_x$ for all $x$. Now define $\vec{c}^0,\vec{c}^1 \in \BB^{\CBE{T}}$ and $\vec{c}' \in \BB^{\CBE{T'}}$ by
\begin{align*}
c^0_x &= \begin{cases}
c''_x, & \textrm{ if $x \neq v$},\\
0, & \textrm{ if $x = v$},
\end{cases}\\
c^1_x &= \begin{cases}
c''_x, & \textrm{ if $x \neq v$},\\
1, & \textrm{ if $x = v$},
\end{cases},\\
c'_x &= c''_x.
\end{align*}
Note that this is well-defined since $\CBE{T'} \subseteq \CBE{T''}$ and $\CBE{T} \subseteq \CBE{T''} \cup \{v\}$. In this formulation we find that
\begin{align}
&\mon{U(T)}[\FF{v} \mapsto \mon{U(T')}][\forall x \in \CBE{T''}\colon F{x} \mapsto c''_x]\nonumber\\
&=  \mon{U_T}[\forall x \in \CBE{T}\colon \FF{x} \mapsto c^1_x] \nonumber\\
&\quad \quad \cdot \mon{U_{T'}}[\forall x \in \CBE{T'}\colon \FF{x} \mapsto c'_x] \nonumber \\
&\quad+\mon{U_T}[\forall x \in \CBE{T}\colon \FF{x} \mapsto c^0_x] \nonumber\\
& \quad \quad \cdot (1-\mon{U_{T'}}[\forall x \in \CBE{T'}\colon \FF{x} \mapsto c'_x])\nonumber\\
&= U_T(\vec{c}^1)U_{T'}(\vec{c}')+U_T(\vec{c}^0)(1-U_T(\vec{c}')). \label{eq:modproof2}
\end{align}
Let $\vec{F}'' \in \BB^{\BE{T''}}$ be the random variable of Definition \ref{def:cbestruc}.2. By our assumption that $V \cap \BE{T'} = \varnothing$, we know that $\BE{T''}$ is the disjoint union of $\BE{T}$ and $\BE{T'}$. Thus we can write $F'' = (\vec{F},\vec{F}')$, with $F$ and $F'$ independent random variables in $\BB^{\BE{T}}$ and $\BB^{\BE{T'}}$, respectively. Furthermore, all paths from $\RRR{T}$ to elements of $\BE{T}$ pass through $\tilde{v}$, and so the random variables $\vec{F}$ and $\struc{T''}(\tilde{v},\vec{F}'',\vec{c}'')$ are independent. It follows that
\begin{align}
&U(T'')(\vec{c}'') \nonumber\\
&= \prob(\struc{T}(\RRR{T''},\vec{F}'',\vec{c}'') = \tone) \nonumber \\
&= \prob(\struc{T''}(\RRR{T''},\vec{F}'',\vec{c}'') = \tone \mid \struc{T''}(\tilde{v},\vec{F}'',\vec{c}'') = \tone)\nonumber \\
& \quad \quad \cdot \prob(\struc{T''}(\tilde{v},\vec{F}'',\vec{c}'') = \tone) \nonumber\\
&\quad + \prob(\struc{T''}(\RRR{T''},\vec{F}'',\vec{c}'') = \tone \mid \struc{T''}(\tilde{v},\vec{F}'',\vec{c}'') = \tzero) \nonumber \\
&\quad \quad \cdot \prob(\struc{T''}(\tilde{v},\vec{F}'',\vec{c}'') = \tzero)\nonumber\\
&= \prob(\struc{T}(\RRR{T},F,\vec{c}^1) = \tone)\cdot \prob(\struc{T'}(\RRR{T'},\vec{F}',\vec{c}') = \tone)\nonumber\\
&\quad \quad + \prob(\struc{T}(\RRR{T},\vec{F},\vec{c}^0) = \tone)\cdot \prob(\struc{T'}(\RRR{T'},\vec{F}',\vec{c}') = \tzero)\nonumber\\
&= U_T(\vec{c}^1)U_{T'}(\vec{c}')+U_T(\vec{c}^0)(1-U_T(\vec{c}')). \label{eq:modproof3}
\end{align}
Combining \eqref{eq:modproof2} and \eqref{eq:modproof3}, we have shown \eqref{eq:modproof1} and the proof is complete. \qed

\subsection*{Proof of Theorem \ref{thm:mainext}}

To prove this result, we prove an extension that is easier to handle via induction. For a node $v$, we write $g_{v,0}$ for the value of $g_v$ at it initialization in Algorithm \ref{alg:SQFU}, i.e., at line 6,9 or 11. Furthermore, if $\gamma(v) \neq \tBE$, let $w_1,\ldots,w_n$ be the elements of $\mathcal{S}_v = \{w \in V \mid \idom(w) = v\}$, in the order in which they are picked in line 15; then $g_{v,i}$ is the value of $g_v$ after $i$ iterations of the loop 14--18. Thus $g_{v,\infty} = g_{v,0}$ if $\gamma(v) = \tBE$, and $g_{v,\infty} = g_{v,n}$ otherwise, in the notation above.

\begin{theorem}
For $v \in V$, define $\mathcal{I}_v = \{w \in V \mid w \prec v \prec \idom(w)\}$.
\begin{enumerate}
\item One has $g_{v,\infty} = \mon{U(T_v[\mathcal{I}_v])}$.
\item Suppose $\gamma(v) \neq \tBE$ and let $w_1,\ldots,w_n$ be the elements of $\mathcal{S}_v$, in the order in which they are picked in line 15. Then
$g_{v,i} = \mon{U(T_v[\mathcal{I}_v \cup \{w_{i+1},\ldots,w_n\}])}$.
\end{enumerate}
\end{theorem}

\begin{proof}
We prove this by induction on $v$, and within a given $v$ by induction on $i$. If $\gamma(v) = \tBE$, then $T_v$ consists of a single BAS with failure probability $\pp(v)$; as such $U(T_v) = \pp(v) = g_v$. This proves the first statement for BEs.

Now suppose $\gamma(v) = \tOR$; the case that $\gamma(v) = \tAND$ is completely analogous. We will prove statement 2, as statement 1 is just the special case $i = n$. We start with the case $i=0$. Since $\mathcal{S}_v \cup \mathcal{I}_v = \{w \in V \mid w \prec v \preceq \idom(w)\}$ and $v \preceq \idom(w)$ for every child $w$ of $v$, we have $\ch(v) \subseteq \mathcal{I}_v \cup \{w_1,\ldots,w_n\}$. It follows that $T' := T_v[\mathcal{I}_v \cup \{w_{i+1},\ldots,w_n\}]$ is a PCFT consisting of $v$, and all its children labeled $\tCBE$. Hence $U(T')(\vec{c}) = \max_{w \in \ch(v)} c_w$ for $\vec{c} \in \BB^{\ch(v)}$, which is represented by the polynomial $1-\prod_{w \in \ch(v)} (1-\FF{w})$. This proves the statement for $i = 0$.

Now suppose the statement is true for a given $i-1$, and let $T_1 = T_v[\mathcal{I} \cup \{w_{i},\ldots,w_n\}]$ and $T_2 = T_v[\mathcal{I} \cup \{w_{i+1},\ldots,w_n\}]$. Then $T_2$ is obtained by replacing the CBE $w_i$ in $T_1$ by the PCFT $T_3 = T_{w_i}[(\mathcal{I}_v \cup \{w_{i+1},\ldots,w_n\}) \cap \recht{desc}(w_i)]$, where $\recht{desc}(w_i)$ denotes the set of descendants of $w_i$. Note that since we pick the $w_j$ in a reverse topological order, all $w_j$ that are descendants of $w_i$ satisfy $j > i$; hence
\[
(\mathcal{I}_v \cup \{w_{i+1},\ldots,w_n\}) \cap \recht{desc}(w_i) = (\mathcal{I}_v \cup \mathcal{S}_v) \cap \recht{desc}(w_i).
\]
We now prove that this set is equal to $\mathcal{I}_{w_i}$. If $w \in (\mathcal{I}_v \cup \mathcal{S}_v) \cap \recht{desc}(w_i)$, then $w \prec w_i$ and $w_i \prec v \preceq \idom(w_i)$; hence $w \in \mathcal{I}_{w_i}$, proving one inclusion. Conversely, if $w \in \mathcal{I}_{w_i}$, then $w_i \prec \idom(w)$; hence $\idom(w_i) \preceq \idom(w)$ by Lemma \ref{lem:idom}. Hence $v \preceq \idom(w_i) \preceq \idom (w)$ and $w \in (\mathcal{I}_v \cup \mathcal{S}_v) \cap \recht{desc}(w_i)$, proving the other conclusion. We conclude that $T_3 = T_{w_i}[\mathcal{I}_{w_i}]$.

Now suppose that $T_3$ and $T_1$ share a BE $w$. If this were the case, then there is a path $v \rightarrow w$ not through $w_i$, namely any such path in $T_2$. Since $w \preceq w_i$ we conclude that $w_i \prec \idom(w)$. But then $w \in \mathcal{I}_{w_i}$; since $T_3 = T_{w_i}[\mathcal{I}_{w_i}]$ this means that $\gamma(w) = \tCBE$ rather than $\tBE$, which is a contradiction. We conclude that no such $w$ exist. Hence we can invoke Lemma \ref{lem:mod} and conclude that
\begin{align*}
g_{v,i} &= g_{v,i-1}[\FF{w_i} \mapsto g_{w,\infty}]\\
&= \mon{U(T_1)}[\FF{w_i} \mapsto \mon{U(T_3)}]\\
&= \mon{U(T_2)},
\end{align*}
as was to be shown.
\end{proof}

\end{document}